\documentclass[sigplan,twocolumn,nonacm]{acmart}

\setcopyright{none}
\renewcommand\footnotetextcopyrightpermission[1]{}

\usepackage{algorithm}
\usepackage{algpseudocode}
\usepackage{amsmath,amsthm}
\usepackage{graphicx}
\usepackage{xcolor}

\theoremstyle{plain}
\newtheorem{theorem}{Theorem}

\theoremstyle{definition}
\newtheorem{definition}[theorem]{Definition}

\theoremstyle{remark}

\begin{document}

\title{Multi-Round Visibility: A Post-Consensus Ordering Layer for DAG-Based BFT}



\author{Pengkun Ren}
\affiliation{%
  \institution{RMIT University}
  \city{Melbourne}
  \state{VIC}
  \country{Australia}
}
\email{s4038427@student.rmit.edu.au}

\author{Hai Dong}
\affiliation{%
  \institution{RMIT University}
  \city{Melbourne}
  \state{VIC}
  \country{Australia}
}
\email{hai.dong@rmit.edu.au}

\author{Zahir Tari}
\affiliation{%
  \institution{RMIT University}
  \city{Melbourne}
  \state{VIC}
  \country{Australia}
}
\email{zahirtari@gmail.com}

\author{Nasrin Sohrabi}
\affiliation{%
  \institution{Deakin University}
  \city{Melbourne}
  \state{VIC}
  \country{Australia}
}
\email{Sohrabi.nasrin.ac@gmail.com}

\renewcommand{\shortauthors}{Ren et al.}

\begin{abstract}
Directed acyclic graph (DAG)-based Byzantine Fault-Tolerant (BFT) protocols achieve high throughput by decoupling dissemination from agreement and allowing many vertices to be committed concurrently. This same concurrency, however, weakens the ordering evidence available at the execution boundary: once units are committed as part of a shared DAG frontier, their final linearization is driven by traversal or deterministic tie-breaking rather than verifiable structural precedence evidence. Prior fair-ordering designs address ordering ambiguity by collecting, validating, or reconstructing stronger transaction-level ordering evidence within the consensus or transaction-processing workflow. While effective for their target semantics, this couples ordering with agreement and places ordering logic on the protocol path.
This paper presents Multi-Round Visibility (MRV), a post-consensus structural ordering layer for DAG-based BFT. MRV reinterprets the committed DAG as an ordering evidence substrate. The key observation is that committed vertices already carry authenticated creator, round, and ancestry metadata, allowing replicas to compute multi-round structural visibility without adding consensus-path messages. MRV accumulates this visibility within a bounded evidence horizon, compares concurrently committed atomic units of fairness (AUFs) only after they coexist in the committed DAG, and derives precedence constraints from one-sided Byzantine-robust visibility advantages. When the committed DAG does not support such a constraint, MRV exposes the remaining ambiguity and resolves it through deterministic graph completion rather than hiding it inside traversal rules.
We implement MRV on a Narwhal/Tusk-based prototype and evaluate it across deployment scales, workloads, and fault settings. Across 5--50 replicas, MRV preserves the high-throughput regime of the underlying DAG-BFT stack, reaching up to 210K TPS with bounded evidence collection and limited throughput impact. These results show that MRV can be integrated as a post-consensus ordering layer while keeping structural-ordering logic off the latency-sensitive consensus path.
\end{abstract}



\maketitle

\section{Introduction}
BFT consensus gives replicated systems a consistent committed history even in the presence of malicious actors \cite{zhang2024reaching}. DAG-based BFT protocols change how that history is assembled and later linearized \cite{wang2023sok}. Instead of a single sequential proposal stream, they construct a causally referenced partial order of concurrently produced units, later deriving a deterministic execution order by locally interpreting the resulting DAG. Recent architectures such as Narwhal and Tusk \cite{danezis2022narwhal} maximize throughput by decoupling dissemination from agreement, allowing large sets of concurrently produced vertices to be proposed and committed concurrently. This architectural gain, however, comes with an ordering cost. While the system must inevitably extend the committed partial order into a deterministic execution sequence, the concurrent nature of the DAG inherently blurs the original propagation history through asynchronous dissemination, cross-round aggregation, and wave-based commitment \cite{mahe2025order}. The challenge, therefore, is not whether such linearization is needed, but what evidence should constrain it once commitment no longer preserves those temporal semantics directly.

This structural ambiguity matters in applications where execution order affects economic or security outcomes. In decentralized exchanges and lending smart contracts \cite{daian2020flash, eskandari2019sok, qin2022quantifying, zhou2021high}, for example, two concurrent operations may be individually valid yet economically non-equivalent depending on which executes first. Such ordering sensitivity is one reason why Maximal Extractable Value (MEV) and transaction-ordering manipulation have become central concerns in blockchain systems \cite{gramlich2024maximal, malkhi2022maximal, baum2022sok, heimbach2022eliminating}. For replicated execution, this creates a design requirement: final order should be constrained by verifiable evidence whenever available, rather than by proposer discretion, traversal conventions, or arbitrary deterministic tie-breaking alone. A growing body of order-fairness work formalizes this requirement by deriving ordering constraints from collective observations of the network \cite{cachin2022quick, kelkar2022order, zhang2020byzantine, kursawe2020wendy, li2024sok, chen2024auncel}.

Existing approaches largely realize fairness by moving ordering logic deeply into the protocol workflow itself. In leader-based systems, Themis \cite{kelkar2023themis}  has the leader collect local transaction orderings, construct a fair proposal, and require replicas to validate that proposal. Pushing this into the DAG setting, DoD \cite{nagda2025dag} embeds fairness computation directly into the transaction-processing pipeline through explicit local-order, global-order, and order-finalization stages; it also requires clients to broadcast transactions to designated workers at every party so that the proposed order can be locally validated. These designs establish important baselines, but they also expose a recurring systems tension: fair-ordering is often incorporated into the protocol path, while making such ordering claims verifiable commonly relies on broad client dissemination assumptions.

Strong transaction-level receive-order fairness is not generally achievable in Byzantine settings without stronger synchrony assumptions \cite{kelkar2023themis}. This motivates a DAG-specific question: \textit{once a common execution slice has been committed, can the system still derive a principled deterministic order from the authenticated structural evidence that remains in the committed DAG?} At that stage, the natural comparison objects are the committed DAG units themselves, since they carry authenticated creator, round, and ancestry information. We refer to these units as atomic units of fairness (AUFs). In Narwhal/Tusk-style systems \cite{danezis2022narwhal}, an AUF corresponds naturally to a committed primary block or certified DAG vertex. When competing transactions reside in different AUFs, constraining the order of those AUFs becomes the post-commit lever for reducing arbitrary execution-order choices.

Our key architectural insight is that, in DAG-BFT systems, this ordering problem can be addressed at the post-commit DAG interpretation layer rather than inside the protocol workflow. Once a sub-DAG is committed, it already contains authenticated causal structure and cross-creator visibility signals that can be reused to constrain ordering. Because visibility is counted at creator granularity, repeated DAG paths do not amplify evidence, and a strong visibility imbalance must be reflected across distinct creators rather than along a single chain. This makes committed-DAG visibility a useful structural signal for post-commit linearization, though not a substitute for global arrival-time observations.

We instantiate this design in MRV, a pluggable ordering layer for DAG-based blockchains that runs after consensus. MRV uses structural evidence from the committed DAG to derive AUF ordering constraints local to each slice, while preserving the underlying consensus protocol and message semantics. It operationalizes this signal conservatively: precedence is introduced only when the committed DAG shows a one-sided visibility advantage after the compared AUFs coexist, while residual ambiguity remains explicit and is resolved deterministically. This placement exposes a favorable systems trade-off: MRV strengthens execution-ordering discipline while preserving the high-throughput consensus path and avoiding separate transaction-level dependency graphs. MRV realizes this design through three core mechanisms.

\textbf{Post-consensus structural interpretation}. MRV treats structural ordering as a deterministic interpretation layer over committed DAG outputs, rather than as part of the consensus path. This placement leaves dissemination, voting, and commit rules unchanged, allowing MRV to build on the safety and liveness guarantees of the base DAG-BFT stack.

\textbf{AUF-level visibility evidence after coexistence}. MRV operates on AUFs rather than individual transactions because AUFs are the committed objects that carry authenticated creator, round, and ancestry metadata. It compares AUF pairs only after both coexist in the committed DAG, so precedence evidence is derived from post-coexistence structural visibility rather than from globally comparable transaction-arrival times.

\textbf{Bounded evidence certification}. MRV accumulates creator-level structural visibility within a bounded post-commit horizon. It introduces a precedence constraint only when the committed DAG exhibits a one-sided Byzantine-robust visibility advantage; otherwise, it abstains and leaves the pair as residual ambiguity for deterministic completion. Thus, deterministic completion produces an executable order without creating additional evidence-backed precedence claims.

The contributions of this paper are summarized as follows:
\begin{itemize}
  \item
  A new architectural design point is identified in which fairness-aware ordering is decoupled from consensus and realized as a pluggable post-consensus interpretation layer over committed DAG structures, preserving the original consensus workflow and its baseline safety and liveness properties.

  \item
  An AUF-level structural evidence model is introduced to extract creator-level visibility evidence directly from the committed DAG, enabling ordering decisions to be based on authenticated structural asymmetry after coexistence rather than on external timing signals.

  \item 
  A conservative ordering mechanism is developed for MRV, where precedence constraints are added only under sufficient committed-DAG support through bounded evidence horizons, maturity thresholds, and explicit abstention in ambiguous cases.

  \item 
  An implementation of MRV over a Narwhal/Tusk prototype and systems evaluation across node scales, batch sizes, fault settings, and geo-distributed environments show that MRV adds bounded post-commit evidence collection with limited throughput overhead.
\end{itemize}

\section{Background}
In BFT consensus, the objective is not merely to agree on a single value, but to establish a consistent execution history for a deterministic state machine replicated across mutually untrusting nodes \cite{schneider1990implementing, xu2023survey, zhang2024reaching}. Classic protocols realize this through State Machine Replication (SMR) \cite{bessani2014state}: all correct replicas execute the same valid requests in the same order. From PBFT \cite{castro1999practical} to modern partially synchronous frameworks such as HotStuff \cite{yin2019hotstuff}, this process has typically been modeled as a leader extending a sequential log. Consequently, traditional BFT systems couple consensus and ordering tightly: the direct protocol output is already a totally ordered history of executable requests.

\subsection{DAG-based BFT and the Concurrency Cost}
While the single-leader architecture provides a natural foundation for sequential execution, its tight coupling of dissemination and agreement can become a scalability bottleneck \cite{wang2023sok}. A leader must both disseminate transaction payloads and drive the cryptographic agreement phases, so throughput is often constrained by the leader's network bandwidth.

DAG-based BFT protocols address this bottleneck by moving toward high-concurrency agreement. Recent designs decouple dissemination and availability from the agreement path, building on a broader line of asynchronous BFT work \cite{guo2020dumbo, gao2022dumbo}. Replicas continuously disseminate transaction batches and broadcast availability-certified vertices attesting that the underlying data are available to a quorum. The consensus layer then orders these lightweight metadata vertices rather than raw transaction payloads. This design pattern appears across DAG-based protocols, from fully asynchronous constructions such as DAG-Rider \cite{keidar2021all} to high-throughput Narwhal-style architectures \cite{danezis2022narwhal}. The dissemination layer also imposes the structure shown in Figure~\ref{fig:dag-bft-primer}. To create a valid vertex at a given round, a replica includes parent references to a quorum of vertices from the preceding round. These mandatory references weave the vertices into a DAG and encode authenticated causal history: if $B \in \mathsf{Anc}(A)$, then the creator of $A$ incorporated $B$ into the ancestry of $A$ before broadcasting it.

\begin{figure}[htbp]
    \centering
    \includegraphics[width=.9\columnwidth]{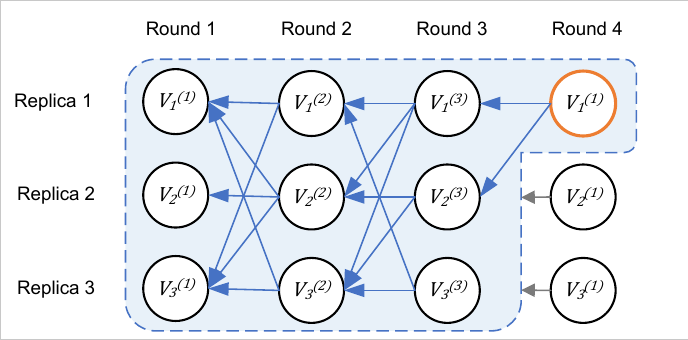}
    \caption{DAG-BFT architecture and committed execution slice. Replicas broadcast availability-certified vertices (circles) that reference a quorum of prior-round vertices (arrows), forming an authenticated causal DAG. Committing a leader vertex (orange) finalizes its previously uncommitted causal history, which forms the committed execution slice linearized by MRV (blue shaded region).}
    \label{fig:dag-bft-primer}
\end{figure}

Consensus protocols built atop this DAG, such as Tusk \cite{danezis2022narwhal} and Bullshark \cite{spiegelman2022bullshark}, operate over this topology rather than a sequential proposal stream. As shown by the shaded region in Figure~\ref{fig:dag-bft-primer}, the protocol periodically commits a leader vertex; this also commits its causal history, the sub-DAG of previously uncommitted vertices referenced by the leader. We call the newly delivered vertices from one such commit step a \textit{committed execution slice}.

This architectural shift changes the interface between consensus and execution. In a linear-log BFT protocol, the consensus output already resembles the execution order. In a DAG-BFT protocol, the natural output of a commit decision is instead a jointly authenticated committed sub-DAG: a set of concurrently produced vertices together with their causal dependencies. The protocol agrees on which vertices enter the shared history, but it does not natively provide a total order over the concurrently committed units. That order must be derived after commitment through traversal rules, deterministic tie-breaking, or an additional ordering layer. This is the concurrency cost of DAG-based BFT. High-throughput dissemination flattens fine-grained arrival signals at the execution boundary, while wave-based commitment exposes slices whose internal precedence is only partially constrained by causality. Deterministic traversal is sufficient for replica agreement, but it provides little explicit evidence for why one causally concurrent committed unit should precede another.

\subsection{Order-Fairness}
This ambiguity in the execution order is not merely a theoretical artifact. In execution-sensitive applications, such as decentralized exchanges and lending protocols, underconstrained ordering can affect economic outcomes and create opportunities for MEV. A robust line of work has therefore sought to establish order-fairness by shifting ordering authority from a single leader or traversal rule to evidence derived from the collective observations of the network.

Because network-wide pairwise receive-order observations may form Condorcet-style cycles, practical systems often weaken exact global ordering goals into graph-based batch-order fairness. The graph used in these schemes is not the underlying consensus structure itself. Instead, systems construct an auxiliary dependency graph in which transactions are vertices and directed edges represent supermajority observations of relative arrival.

The practical realization of this model reveals a clear architectural evolution. Aequitas \cite{kelkar2020order} pioneered this graph-based foundation, circumventing Condorcet cycles by grouping strongly connected components (SCCs) into unordered batches, though it suffered from weak liveness because cyclic dependencies could indefinitely delay finalization. Themis \cite{kelkar2023themis} addressed this liveness limitation through deferred ordering and liveness anchors, resulting in a monolithic architecture that constructs $\mathcal{O}(|B|^2)$ transaction dependency graphs within the consensus path. SpeedyFair \cite{mu2024separation} later introduced a decoupled architecture, moving fair-ordering into an optimistic parallel pipeline to hide ordering latency from the core BFT agreement. Across these designs, fairness evidence is still collected and maintained as a transaction-level ordering structure distinct from the base protocol output.

When the traditional order-fairness paradigm intersects with modern DAG-BFT architectures, a recurring systems tension emerges. DAG-BFT protocols achieve high throughput precisely by decoupling transaction dissemination from the agreement path. Pushing transaction-level fair ordering into this setting, as in DoD \cite{nagda2025dag}, preserves the graph-based formulation but introduces explicit local-order, global-order, and finalization stages into the transaction-processing workflow. In addition to this computational mismatch, making transaction-level fairness claims checkable often relies on broad client dissemination \cite{kelkar2023themis, nagda2025dag}, so that fairness-relevant inputs are consistently observed across parties. This assumption creates tension with the core DAG-BFT design goal: it introduces a system-wide dissemination burden alongside the communication efficiency originally gained by decoupling dissemination from agreement.

These computational and communication costs expose a linearization gap. A DAG-BFT protocol already outputs a jointly authenticated committed sub-DAG: its vertices carry round numbers, creator identities, and authenticated ancestor relations, recording cross-creator structural visibility among committed units. This raises a systems opportunity: instead of carrying a separate ordering-evidence structure through dissemination, can the execution layer reuse the committed DAG itself as ordering evidence? MRV builds on this opportunity by interpreting committed-DAG visibility as a post-consensus signal for constructing a slice-local execution order without changing the underlying protocol.

\section{System Model and Problem Statement}
\label{sec:model}

We model MRV as a post-consensus interpretation layer over a DAG-based BFT stack. MRV does not modify the underlying consensus message flow, voting rules, or commit conditions. Instead, it acts as a deterministic bridge between consensus and execution: it operates on the growing committed DAG structure and the ordered committed execution slices exported by the base stack, and for each slice outputs a deterministic, slice-local order over its AUFs. The following subsections formalize these exported objects and the assumptions MRV makes about them.

\subsection{System Setting and Trust Assumptions}
We consider a system of $n$ replicas, denoted by $\mathcal{R}=\{1,\dots,n\}$, which also serve as the creators of DAG vertices. The base protocol authenticates creator identities and parent references through its certification and signature machinery, making committed DAG metadata publicly verifiable by all correct replicas. MRV is defined over an abstract committed-DAG exporter interface: it assumes only that the base stack eventually exposes a common sequence of finalized committed outputs to all correct replicas, under the synchrony and liveness assumptions of that base protocol instance. MRV itself introduces no additional synchrony requirement.

\subsection{Threat Model}
\label{sec:threat}

We inherit the standard authenticated Byzantine fault model used by DAG-based BFT protocols such as Narwhal/Tusk, DAG-Rider, and Bullshark~\cite{danezis2022narwhal,keidar2021all,spiegelman2022bullshark}. Among the $n$ replicas, at most $f$ are Byzantine, with $n\ge 3f+1$. The adversary is computationally bounded but may coordinate faulty replicas and schedule messages within the communication assumptions of the base protocol. Byzantine replicas may deviate during DAG construction by omitting vertices, delaying or withholding their own messages, choosing parent references strategically, or attempting equivocation.

MRV operates only on the canonical committed DAG exported by the base stack. We assume the base protocol exports certified vertices whose creator identities and parent references are authenticated, together with the same committed DAG prefix $\mathcal{G}^c_R$ and the same ordered slice sequence $(S_1,S_2,\ldots)$ to all correct replicas. As part of this exporter contract, for each creator-round pair $(c,t)$, the exporter exposes at most one canonical committed AUF $Y_{c,t}$; non-canonical equivocation attempts are not part of MRV's input and are excluded from creator-level visibility counts. This matches the standard interface of quorum-certified DAG-BFT stacks, where equivocation handling and canonicalization belong to the base protocol's validity, certification, and delivery rules rather than MRV's post-consensus interpretation layer.

At the MRV layer, Byzantine replicas and adversarial scheduling can affect the structural evidence that appears in $\mathcal{G}^c_R$, including visibility trajectories, maturity outcomes, and conflicting directional signals induced by strategic parent selection or message scheduling. MRV treats this committed structure as its evidence surface: evidence-backed precedence is derived only from the structural visibility predicate in Section~\ref{sec:fairness_target}, while pairs not supported by that predicate remain residual ambiguity for deterministic completion.

\subsection{Committed DAG and Execution Slices}
The base DAG-BFT protocol outputs a monotonically growing committed DAG view. Let $R$ denote the current committed frontier (e.g., the highest committed round). We define $\mathcal{G}^c_R = (V^c_R, E^c_R)$ as the common committed DAG view available at frontier $R$. All correct replicas eventually agree on the identical $\mathcal{G}^c_R$ for any given $R$.

For any committed vertex $X \in V^c_R$, the DAG provides authenticated structural metadata, including its creator $\mathsf{cr}(X) \in \mathcal{R}$, strictly increasing round number $r(X)$, parent references $\mathsf{Par}(X)$, and reflexive ancestor closure $\mathsf{Anc}(X)$ (by convention, $X \in \mathsf{Anc}(X)$).

We abstract the base DAG-BFT stack as exposing an agreed, ordered sequence of committed execution slices $S_1,S_2,\dots$. All correct replicas observe the same slice sequence in the same order. Each slice consists of the newly delivered AUFs associated with one commit decision. MRV does not require the base protocol to natively name such slices; it only requires an adapter that exposes the newly delivered AUFs for each commit decision before execution-layer traversal is applied. In a Narwhal/Tusk-style instantiation, this corresponds to the committed leader's causal history after removing AUFs that have already appeared in earlier slices. Because slices contain only newly delivered AUFs, they are mutually disjoint, so that each AUF belongs to exactly one slice. After removing previously delivered AUFs, a slice need not be causally closed by itself; MRV treats the slice strictly as the output domain, while the committed DAG prefix remains available as the evidence domain. MRV is invoked independently on each slice $S \subseteq V^c_R$. To derive the order for a given slice $S$, MRV need not wait for the full future committed DAG; it only consults the committed prefix available up to the frontier required by its bounded observation horizon.

\subsection{Atomic Units of Fairness and Structural Evidence}
MRV orders AUFs within a single committed execution slice. We abstract an AUF as the committed DAG vertex that carries the creator, round, and ancestry metadata used by MRV; in Narwhal/Tusk-style architectures, this naturally maps to a committed primary vertex. The order of transactions within an AUF and the order among committed slices are inherited from the underlying batching and consensus layers. Structural metadata are therefore attached to AUFs rather than to individual transactions. Let $X, A, B \in S$ denote AUFs within the same slice.

To evaluate structural precedence, MRV extracts cross-creator visibility from $\mathcal{G}^c_R$. Let $Y_{c,t}$ denote the canonical committed AUF exposed by the exporter for creator $c$ at round $t\le R$, if it exists. The structural visibility count of $X$ at round $t \ge r(X)$ is defined as:
$C_X(t) = \#\{c \in \mathcal{R} \mid Y_{c,t} \text{ exists and } X \in \mathsf{Anc}(Y_{c,t})\}.$ MRV limits post-commit evidence collection with a system-wide observation cap $W_{\max}$. For each AUF $X$, its stopping time $h_X$ is defined as: $h_X := \min \big( \{ t \ge r(X) \mid C_X(t) \ge 2f+1 \} \cup \{ r(X) + W_{\max} \} \big).$ The maturity indicator records whether $X$ reached quorum visibility within this cap: $\mathsf{mature}(X) = \mathtt{true} \iff C_X(h_X) \ge 2f+1. $

For pairwise alignment, MRV delays the evaluation of pair $(A, B)$ until both AUFs have reached their stopping times. We define the coexistence-alignment round $s(A,B) = \max(r(A), r(B))$, and the pair horizon $H(A,B) = \max(h_A, h_B)$. A slice $S$ reaches its final slice sealing time at $T(S) = \max_{X \in S} h_X$.

\subsection{Fairness Target: Structural Visibility Precedence}
\label{sec:fairness_target}

MRV defines an AUF-level structural ordering target over committed DAG outputs. Its evidence source is authenticated creator, round, and ancestry metadata in the committed DAG, not replica-local transaction arrival time. Accordingly, MRV provides a structural guarantee: it certifies an AUF-level precedence relation only when the committed DAG satisfies the structural visibility predicate, and abstains otherwise.

For two AUFs $A,B \in S$, define their visibility delta at round $t$ as $\Delta(A,B,t)=C_A(t)-C_B(t).$ The post-coexistence observation window for $(A,B)$ is $\mathcal{W}(A,B)=\{s(A,B)+1,\dots,H(A,B)\}.$ If $H(A,B)\le s(A,B)$, this window is empty and no post-coexistence structural precedence is inferred.

\begin{definition}[Structural Visibility Precedence]
For a pair $(A,B)$ in the same committed execution slice $S$, we say $A$ has structural visibility precedence over $B$, denoted $A \triangleright_{\mathrm{SVP}} B$, if (1) both $A$ and $B$ are mature; (2) there exists a round $t \in \mathcal{W}(A,B)$ such that $\Delta(A,B,t)\ge f+1$; and (3) there does not exist a round $t \in \mathcal{W}(A,B)$ such that $\Delta(B,A,t)\ge f+1$.
\end{definition}

This predicate is MRV's evidence-backed precedence relation. Threshold $f+1$ is applied to creator-level visibility deltas, so a positive signal must be reflected across distinct creators rather than amplified by repeated paths through the DAG.

\begin{definition}[Conditional AUF-Level Structural Fairness]
An output order for a committed slice satisfies conditional AUF-level structural fairness if it preserves every evidence-backed structural visibility precedence relation that remains compatible with the slice's hard causal order and the graph-level conflict resolution defined by MRV. Pairwise relations without mature, one-sided, non-conflicting structural visibility evidence remain residual ordering ambiguity and are completed deterministically without being counted as evidence-backed precedence.
\end{definition}

Section~\ref{sec:design} defines enforceability via the MRV precedence graph, where causal and SVP edges are assembled and only SVP constraints that survive SCC condensation are preserved as structural precedence claims.

This target follows the systems principle behind order-fairness: execution order should be constrained by verifiable evidence rather than proposer discretion or arbitrary tie-breaking. MRV applies this principle to a DAG-native evidence substrate, using committed creator-level visibility to constrain slice-local AUF linearization.

\begin{figure*}[t]
  \centering
  \includegraphics[width=0.9\textwidth]{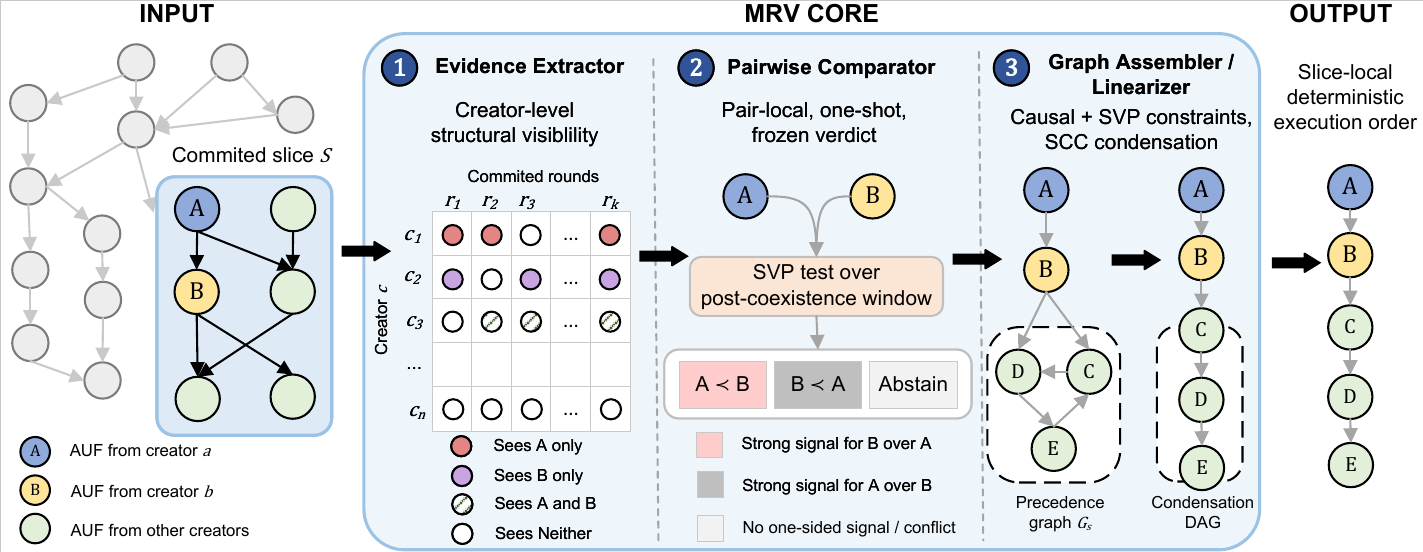}
    \caption{Execution pipeline of MRV. Operating post-consensus, MRV takes a committed execution slice $S$ as read-only input and derives a deterministic slice-local order through three local stages. (1) The Evidence Extractor accumulates creator-level structural visibility for each AUF over bounded committed rounds until a quorum threshold or observation cap is reached. (2) The Pairwise Comparator evaluates mature AUF pairs over a post-coexistence window and freezes a verdict only when the committed DAG provides unconflicted, one-sided strong evidence. (3) The Graph Assembler builds the precedence graph $G_S$ from hard causal constraints and evidence-backed verdicts, then condenses SCCs, topologically orders the DAG, and completes residual ties deterministically.}
  \label{fig:mrv-overview}
\end{figure*}

\subsection{Problem Statement}
\label{sec:problem}

For each committed execution slice $S$, MRV derives a deterministic slice-local total order $\prec_S$ over the AUFs in $S$. Its fairness target is the conditional AUF-level structural fairness defined in Section~\ref{sec:fairness_target}: when the committed DAG provides mature, one-sided structural visibility evidence for a same-slice AUF pair, MRV should preserve the corresponding enforceable precedence relation in the final slice order. Pairs without such support remain residual ordering ambiguity; MRV completes them deterministically while keeping the resulting order separate from evidence-backed precedence.

\textbf{Slice Boundary.}
MRV operates within the linearization boundary of a committed execution slice. Cross-slice reactive strategies, such as placing a transaction in $S_{i+1}$ in response to one committed in $S_i$, exercise a different control point: transaction admission, mempool scheduling, or the base commit sequence. MRV is complementary to these mechanisms and targets the intra-slice ordering ambiguity exposed by concurrent DAG commitment.

Given a monotonically growing committed DAG view $\mathcal{G}^c_R$, a committed execution slice $S$, and a common observation cap $W_{\max}$, MRV produces a slice-local total order satisfying:

\noindent\textbf{(1) Deterministic Agreement:} All correct replicas derive the same $\prec_S$ from the same committed inputs.\\
\noindent\textbf{(2) Causal Consistency:} $\prec_S$ extends the committed DAG's causal partial order within $S$.\\
\noindent\textbf{(3) Slice Locality:} MRV orders only AUFs in $S$, while its evidence may use committed descendants outside $S$ within the bounded observation horizon.\\
\noindent\textbf{(4) Evidence Conservatism:} MRV certifies an evidence-backed precedence constraint only from mature, one-sided, non-conflicting structural visibility evidence.\\
\noindent\textbf{(5) Conditional Structural Fairness:} Every enforceable SVP relation from Section~\ref{sec:fairness_target} is preserved in $\prec_S$.\\
\noindent\textbf{(6) Bounded Completion:} $\prec_S$ is sealed once the committed frontier reaches $T(S)=\max_{X\in S}h_X$, so post-commit evidence collection is bounded by $W_{\max}$.

\section{MRV Design and Workflow}
\label{sec:design}

\subsection{Overview and Execution Pipeline}
MRV is a deterministic post-consensus interpretation pipeline executed independently by all correct replicas. It takes as input the monotonically growing committed DAG view $\mathcal{G}^c_R$, the current committed execution slice $S$, and the system-wide observation cap $W_{\max}$, and produces a deterministic, slice-local total order $\prec_S$ over the AUFs in $S$. An overview of the MRV architecture is shown in Fig.~\ref{fig:mrv-overview}.

MRV separates the ordering domain from the evidence domain. The ordering domain is restricted to the AUFs contained in the current slice $S$. The evidence domain may extend beyond $S$: MRV may consult committed descendants outside $S$ from later rounds to extract structural visibility evidence, provided that they lie within the bounded observation horizon. MRV is invoked independently for each committed execution slice; cross-slice delivery order remains the order exported by the base consensus layer.

MRV realizes the ordering of each slice through five logical stages. As the committed frontier advances, these stages are incrementally interleaved in an online manner:

\noindent\textbf{1. Structural Visibility Accumulation:} As new vertices are committed, MRV incrementally tracks the creator-level visibility of each AUF $X \in S$ across subsequent rounds.
\\
\noindent\textbf{2. AUF Stopping Times:} Each AUF reaches a stopping time $h_X$ either by achieving quorum-grade visibility or by hitting the common observation cap $W_{\max}$.
\\
\noindent\textbf{3. Pair Horizons and Frozen Verdicts:} For each pair $(A, B) \in S$, MRV waits until the pair horizon $H(A,B)$ is reached, extracts pairwise precedence evidence from post-coexistence rounds, and freezes the resulting verdict.
\\
\noindent\textbf{4. Slice Sealing and Precedence Graph Construction:} Once the committed frontier reaches $T(S)$, all AUFs in $S$ have reached their stopping times. MRV then constructs a slice-local precedence graph from hard causal constraints and evidence-backed pairwise verdicts.
\\
\noindent\textbf{5. Graph Linearization and Deterministic Completion:} MRV compresses cycles through SCC condensation, topologically orders the condensation DAG, and applies the deterministic completion key $\kappa$ only to residual ambiguity not fixed by causal or enforceable evidence-backed constraints.

MRV first determines which AUFs have accumulated structural evidence, then identifies which AUF pairs admit evidence-backed precedence, and finally turns the resulting constrained partial order into an executable total order. Deterministic completion is used only for residual ambiguity; it does not create additional structural precedence claims.

\subsection{Structural Visibility, Observation Horizon, and AUF Maturity}
Building on the structural visibility count $C_X(t)$ defined in Section~\ref{sec:model}, MRV uses committed descendants to incrementally accumulate ordering evidence for each AUF over rounds.

\textbf{Creator-Level Visibility Extraction.} For an AUF $X \in S$ and round $t \ge r(X)$, visibility is measured by $C_X(t) = \#\{c \in \mathcal{R} \mid Y_{c,t} \text{ exists and } X \in \mathsf{Anc}(Y_{c,t})\}.$ This quantity counts the number of distinct creators whose canonical committed AUFs include $X$ in their ancestor closure. Because the metric is creator-level and ancestry-based, it does not depend on replica-local receive order.

\textbf{Observation Cap and Stopping Time.}
MRV treats an AUF as having accumulated quorum-grade structural visibility once $C_X(t) \ge 2f+1$. Since at most $f$ creators are Byzantine, this means that the AUF appears in the authenticated ancestry of at least $f+1$ correct creators at round $t$; this is a structural visibility witness, not a claim about transaction-level arrival order. To avoid waiting indefinitely for delayed AUFs, MRV bounds evidence accumulation by the observation cap $W_{\max}$. The cap is a protocol configuration parameter that controls the trade-off between evidence coverage and post-commit ordering delay. A larger $W_{\max}$ gives MRV more committed rounds from which to observe structural visibility, while a smaller $W_{\max}$ closes evidence windows earlier and may increase immature AUFs, abstentions, and deterministic completion. In our prototype, this cap is instantiated using the configured committed-history retention window passed to the MRV executor.

Using the stopping time already defined in Section~\ref{sec:model},
\[
h_X := \min \big( \{ t \ge r(X) \mid C_X(t) \ge 2f+1 \} \cup \{ r(X) + W_{\max} \} \big),
\]
MRV replaces an unbounded hindsight window with a bounded online stopping rule driven entirely by committed-DAG structure. As illustrated in Figure~\ref{fig:mrv-visibility}, for each AUF in the active slice, the evidence window closes deterministically at the first committed round where the AUF either reaches the quorum visibility threshold or reaches the observation cap.

\textbf{Maturity Semantics.} Based on $h_X$, the maturity indicator $\mathsf{mature}(X)$ is true iff $C_X(h_X) \ge 2f+1$. An AUF capped before reaching this threshold is immature. Only mature AUFs are eligible to induce positive precedence claims in pairwise comparison. Immature AUFs still appear in the final total order, but they induce no evidence-backed precedence constraints.

\begin{figure}[htbp]
    \centering
    \includegraphics[width=\columnwidth]{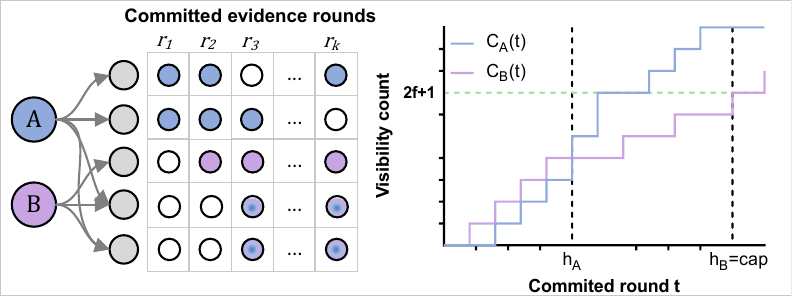}
    \caption{Structural visibility accumulation and AUF stopping times. Matrix cells indicate whether creator-round AUFs see $A$, $B$, both, or neither. MRV counts distinct creators that see each target AUF and fixes $h_X$ at the $2f+1$ maturity threshold or the observation cap $r(X)+W_{\max}$.}
    \label{fig:mrv-visibility}
\end{figure}

\subsection{Pairwise Verdicts from Post-Coexistence Evidence}
Given the stopping time $h_X$ and maturity indicator $\mathsf{mature}(X)$ of each AUF, MRV converts per-AUF visibility trajectories into pairwise precedence verdicts. Each pair is frozen once its pair horizon is reached, so later committed rounds cannot rewrite an already-issued pairwise claim. This pairwise comparator is the only stage at which MRV introduces explicit precedence claims; all later stages merely assemble or complete those claims deterministically.

\textbf{Mature-Pair Gate.} For any pair $(A, B) \in S$, if either $\mathsf{mature}(A)$ or $\mathsf{mature}(B)$ evaluates to $\mathtt{false}$, MRV abstains. Immature AUFs provide truncated evidence trajectories and therefore do not support explicit positive precedence claims.

\textbf{Post-Coexistence Alignment.} For mature pairs, MRV compares visibility only after both AUFs coexist in the committed history. With coexistence-alignment round $s(A,B) = \max(r(A), r(B))$, MRV evaluates the visibility delta over the strictly post-coexistence window $t \in \{s(A,B)+1, \dots, H(A,B)\}$ as $\Delta(A,B,t) = C_A(t) - C_B(t)$.
This alignment excludes pre-coexistence lead from being interpreted as structural precedence evidence. The window starts at $s(A,B)+1$ because $s(A,B)$ is the later creation round of the pair; only subsequent committed rounds can provide creator-level visibility observations that may reflect both AUFs.

\textbf{One-Sided Strong Signal Rule.} MRV implements the SVP predicate from Section~\ref{sec:fairness_target} using a visibility-delta threshold of $f+1$. Since at most $f$ creators are Byzantine and each creator contributes at most one visibility count per round, a visibility advantage of $f+1$ cannot be generated solely by Byzantine creators. This does not imply transaction-level arrival priority; it identifies a Byzantine-robust structural signal in the committed DAG. Let $pos(A,B)$ and $neg(A,B)$ indicate whether the post-coexistence window contains a positive or negative strong signal, respectively. MRV returns $\mathtt{Edge}(A\to B)$ iff $pos(A,B)$ is true and $neg(A,B)$ is false. If both signals appear, the DAG exhibits conflicting evidence and MRV abstains; if neither appears, MRV abstains for lack of signal. MRV relies on the existence of an unopposed strong signal rather than an accumulated score, preventing later descendants from disproportionately amplifying an early advantage. Figure~\ref{fig:mrv-pairwise} illustrates this rule over the post-coexistence visibility-delta trajectory.

\textbf{Verdict Taxonomy.} Based on these rules, Algorithm~\ref{alg:pair_verdict} evaluates each pair exactly once at $R \ge H(A,B)$, yielding one of four verdict categories (where $\mathtt{Edge}$ may be instantiated as either $\mathtt{Edge}(A \to B)$ or $\mathtt{Edge}(B \to A)$): 
\textbf{Edge}, if a one-sided strong signal is observed; 
\textbf{Abstain-Truncated}, if at least one AUF fails to mature; 
\textbf{Abstain-Conflict}, if both positive and negative strong signals are observed; and 
\textbf{Abstain-NoSignal}, if no strong directional signal is observed.

Abstention is MRV's certification boundary. An abstained pair is left outside the evidence-backed precedence relation and remains residual ordering ambiguity in $G_S$; deterministic completion resolves this ambiguity only to obtain an executable total order. This separation lets MRV distinguish precedence supported by committed-DAG evidence from ordering choices made solely for deterministic executability. As a result, abstention reduces evidence coverage rather than evidence soundness: it may increase the portion of the final order determined by completion, but it does not turn unsupported evidence into a structural precedence claim.

\begin{figure}[htbp]
    \centering
    \includegraphics[width=.9\columnwidth]{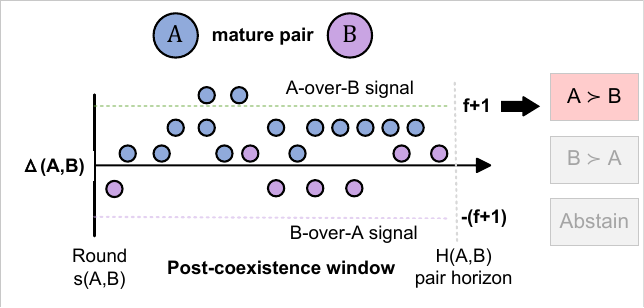}
    \caption{Pairwise comparison and verdict extraction. For a mature pair $(A,B)$, MRV evaluates visibility deltas only after both AUFs coexist. An unopposed $f+1$ crossing yields $A\to B$; conflict or no strong signal yields abstention.}
    \label{fig:mrv-pairwise}
\end{figure}

\begin{algorithm}
\caption{MRV Pairwise Verdict for $(A,B)$}
\label{alg:pair_verdict}
\begin{algorithmic}[1]
\Require Pair $(A,B)$, stopping times, maturity flags, and visibility values up to $H(A,B)$
\If{$\neg\mathsf{mature}(A) \lor \neg\mathsf{mature}(B)$}
    \State \Return $\mathtt{Abstain\text{-}Truncated}$
\EndIf
\State $s \gets \max(r(A),r(B))$, $H \gets \max(h_A,h_B)$
\State $pos \gets \mathsf{false}$, $neg \gets \mathsf{false}$
\For{$t=s+1$ \textbf{to} $H$}
    \State $\Delta \gets C_A(t)-C_B(t)$
    \State $pos \gets pos \lor (\Delta \ge f+1)$
    \State $neg \gets neg \lor (\Delta \le -(f+1))$
\EndFor
\If{$pos \land \neg neg$}
    \State \Return $\mathtt{Edge}(A \to B)$
\ElsIf{$neg \land \neg pos$}
    \State \Return $\mathtt{Edge}(B \to A)$
\ElsIf{$pos \land neg$}
    \State \Return $\mathtt{Abstain\text{-}Conflict}$
\Else
    \State \Return $\mathtt{Abstain\text{-}NoSignal}$
\EndIf
\end{algorithmic}
\end{algorithm}

\subsection{Slice Sealing, Precedence Graph Construction, and Linearization}
While Algorithm~\ref{alg:pair_verdict} establishes pairwise constraints, MRV assembles these local verdicts into an executable total order.

\textbf{Slice Sealing.} MRV defines the slice sealing time for a slice $S$ as $T(S) = \max_{X \in S} h_X.$ Since the pair horizon is $H(A,B) = \max(h_A, h_B)$, it follows that $H(A,B) \le T(S)$ for every pair $(A,B) \in S \times S$. Therefore, once the committed frontier reaches $R \ge T(S)$, every AUF in $S$ has reached its stopping time and every pairwise verdict is frozen.

\textbf{Precedence Graph Construction.} At this sealing horizon, MRV constructs the slice-local precedence graph $G_S=(S,E_S)$ with two classes of edges. The first class consists of hard causal constraints: $E_S^{\mathsf{causal}} = \{(B,A) \mid A,B\in S \land B\in\mathsf{Anc}(A)\}.$ These edges require the final slice order to extend the committed DAG's causal partial order. The second class consists of evidence-backed structural precedence edges: $E_S^{\mathsf{svp}} = \{(A,B) \mid \text{the comparator returns } \mathtt{Edge}(A\to B)\}.$ MRV sets $E_S=E_S^{\mathsf{causal}}\cup E_S^{\mathsf{svp}}$. Any abstention yields no evidence-backed edge in $E_S^{\mathsf{svp}}$.

\textbf{Graph Linearization and Deterministic Completion.} Because pairwise constraints are derived independently from local structural visibility, they need not be perfectly transitive. Independently justified edges may therefore form precedence cycles, making a graph-level assembly phase necessary. MRV identifies SCCs in $G_S$ and condenses them into a Directed Acyclic Graph (DAG) of SCCs, which is then linearized by deterministic topological ordering. 

Causal constraints represent the committed DAG's partial order and are not treated as fairness claims. MRV's structural fairness claims are derived solely from evidence-backed comparator edges in $E_S^{\mathsf{svp}}$. If evidence-backed edges create cycles, MRV treats the affected evidence as graph-level ambiguity: only constraints that survive SCC condensation are claimed as enforceable structural precedence. Causal constraints remain hard execution constraints throughout linearization. Within each SCC, MRV topologically orders the subgraph induced by $E_S^{\mathsf{causal}}$ and applies the deterministic key $\kappa$ only among causally incomparable AUFs. Therefore, deterministic completion never violates causality or turns residual ordering choices into structural fairness claims.

To make completion replica-consistent, MRV introduces an abstract deterministic key $\kappa(X)$ that induces the same total order over AUFs at all correct replicas. The key is used in two places. First, for an SCC $C$, MRV lifts the key as $\kappa(C)=\min_{X\in C}\kappa(X)$ and uses it to break ties among simultaneously eligible SCCs in the condensation DAG. Second, within an SCC, $\kappa(X)$ is used only after the causal subgraph is enforced, as described above. In our prototype, $\kappa(X)$ is the lexicographic tuple $(r(X),\mathsf{cr}(X),\mathsf{dig}(X))$, where $\mathsf{dig}(X)$ denotes the canonical digest of $X$. After SCC condensation, topological ordering, and deterministic completion, the final sequence $\prec_S$ is delivered to the execution engine.

\subsection{Online Replica Procedure and State Management}
MRV is implemented as an incremental online procedure driven by the advancing committed frontier. This procedure demonstrates that MRV avoids offline hindsight reconstruction; instead, it updates ordering evidence incrementally as the base protocol advances the committed frontier. Rather than deferring all computation until a slice seals, each replica maintains bounded active state for every unsealed slice and updates that state whenever a new committed round becomes available.

For each active AUF $X$, MRV stores (i) the per-round structural visibility profile $\{C_X(t)\}$ over the rounds observed so far within the relevant horizon, (ii) the stopping time $h_X$ if already determined, and (iii) the maturity flag $\mathsf{mature}(X)$. For each active pair $(A,B)$, MRV stores whether the pairwise verdict has already been frozen and, if so, the frozen outcome returned by Algorithm~\ref{alg:pair_verdict}. For each active slice $S$, MRV stores the member AUFs together with the intermediate data needed to construct $G_S$ once the slice seals.

When a new committed frontier round $R$ is delivered, the replica first materializes the canonical committed AUFs $\{Y_{c,R}\}_c$, if any. For every active AUF $X$, MRV initializes the current-round visibility count $C_X(R)$ to zero. It then performs a bounded backward traversal from each canonical AUF $Y_{c,R}$ through its committed ancestors and increments $C_X(R)$ for every active AUF $X$ encountered. This traversal is restricted to active AUFs within the relevant bounded horizon, rather than the entire historical DAG. This yields the creator-level structural visibility counts for round $R$ without relying on local receive order or external timing signals.

These round-$R$ counts are then used to settle AUF stopping times. If $h_X$ is still undefined and $C_X(R)\ge 2f+1$, MRV fixes $h_X := R$ and marks $X$ as mature. If $h_X$ remains undefined and $R \ge r(X)+W_{\max}$, MRV caps $X$ at $h_X := r(X)+W_{\max}$ and marks it as immature. Fixing $h_X$ does not terminate visibility tracking for $X$: MRV continues to record later $C_X(t)$ values until every pair involving $X$ has frozen, because pairwise comparison may extend to a later horizon $H(A,B)$.

Once both $h_A$ and $h_B$ are known, the pair horizon $H(A,B)$ becomes known as well. Whenever the frontier satisfies $R \ge H(A,B)$, the pair $(A,B)$ is evaluated exactly once using Algorithm~\ref{alg:pair_verdict}, and its verdict is frozen permanently.

Finally, when $R \ge T(S)$ for active slice $S$, every AUF in $S$ has reached its stopping time and every pairwise verdict is frozen. MRV then constructs the precedence graph $G_S$, performs SCC condensation and deterministic topological linearization, applies the completion key $\kappa$ to resolve residual ambiguity, outputs the final order $\prec_S$, and releases all state for $S$. Because evidence accumulation is bounded by $W_{\max}$ and each pair is frozen at most once, the online procedure maintains bounded active state and avoids unbounded retrospective dependence on the committed DAG.

\begin{algorithm}
\caption{MRV Processing at Committed Frontier $R$}
\label{alg:online_exec}
\begin{algorithmic}[1]
\Require New committed frontier round $R$, canonical AUFs $\{Y_{c,R}\}$, active slices, and MRV state
\State Initialize $C_X(R)\gets 0$ for each active AUF $X$
\For{each canonical AUF $Y_{c,R}$}
    \For{each active AUF $X \in \mathsf{Anc}(Y_{c,R})$ reached by bounded traversal}
        \State $C_X(R) \gets C_X(R)+1$
    \EndFor
\EndFor
\For{each active AUF $X$ with $h_X$ undefined}
    \If{$C_X(R)\ge 2f+1$}
        \State $h_X \gets R$; $\mathsf{mature}(X)\gets \mathtt{true}$
    \ElsIf{$R\ge r(X)+W_{\max}$}
        \State $h_X \gets r(X)+W_{\max}$; $\mathsf{mature}(X)\gets \mathtt{false}$
    \EndIf
\EndFor
\For{each unfrozen pair $(A,B)$ with $h_A,h_B$ defined and $R\ge H(A,B)$}
    \State Freeze $\mathsf{verdict}(A,B)\gets \Call{PairwiseVerdict}{A,B}$
\EndFor
\For{each active slice $S$ with all $h_X$ defined and $R\ge T(S)$}
    \State Build $G_S$ from causal edges and frozen $\mathtt{Edge}$ verdicts
    \State Condense SCCs and topologically order the SCC-DAG using $\kappa(C)$
    \State Order each SCC by causal constraints, then $\kappa(X)$
    \State Output $\prec_S$ and release state for $S$
\EndFor
\end{algorithmic}
\end{algorithm}

\subsection{Adversarial Influence and Conservative Failure Mode}
\label{sec:adversarial_influence}

MRV treats Byzantine behavior as part of the committed-DAG evidence surface. Faulty creators may influence structural visibility through their own parent choices, omissions, equivocations handled by the base exporter, or message scheduling effects permitted by the base protocol. However, MRV counts visibility at creator granularity, and each creator contributes at most one canonical visibility observation per round. Therefore, Byzantine-only contributions can account for at most $f$ units of a directional visibility margin in any round; an $f+1$ margin necessarily extends beyond Byzantine creator contributions. The effect of adversarial behavior is captured by MRV's verdict taxonomy. Suppressed or delayed visibility may prevent maturity and lead to $\mathtt{Abstain\text{-}Truncated}$; opposing strong visibility patterns lead to $\mathtt{Abstain\text{-}Conflict}$; weak or insufficient directional evidence leads to $\mathtt{Abstain\text{-}NoSignal}$. In each case, the pair remains outside $E_S^{\mathsf{svp}}$ and is resolved only as residual ambiguity during deterministic completion. Thus, adversarial influence can reduce MRV's evidence coverage, but it does not change the rule by which MRV certifies structural precedence edges.

\subsection{Complexity, State Bound, and Practical Cost Drivers}
\label{sec:complexity}

Algorithm~\ref{alg:online_exec} adds bounded local overhead above the base DAG-BFT stack. MRV introduces no consensus-path messages and no extra cryptographic verification; its cost comes from bounded backward traversal for visibility updates, pairwise verdict evaluation over AUFs in a slice, and SCC condensation/topological linearization once the slice seals.

For a slice of size $|S|$, a direct implementation performs $O(|S|^2W_{\max})$ pairwise verdict work and $O(|S|^2)$ slice-local graph assembly in the worst case. The quadratic factor is over AUFs, not individual transactions; transaction ordering inherits the AUF order through containment. Active per-slice state is bounded by visibility profiles over at most $W_{\max}$ rounds together with $O(|S|^2)$ pair verdict state, and all state for $S$ is released once $R\ge T(S)$. Thus, MRV bounds the additional committed-round evidence horizon by $W_{\max}$; progress of the committed frontier remains the responsibility of the base protocol. A fuller cost breakdown is provided in the supplementary material.


\subsection{Analysis and Guarantees}
\label{sec:analysis}

We analyze MRV as a deterministic post-consensus interpretation layer over the exporter interface in Section~\ref{sec:model}. The base DAG-BFT stack provides a common committed DAG prefix and a common ordered sequence of committed execution slices; MRV deterministically maps these exported objects to slice-local execution orders. We summarize the main guarantees below and provide full proofs in the supplementary material.


\begin{theorem}[Deterministic Agreement]
\label{thm:deterministic_agreement}
For any committed execution slice $S$, if two correct replicas observe the same exported committed prefix up to $R\ge T(S)$, the same $W_{\max}$, and the same completion key $\kappa$, then they output the same total order $\prec_S$.
\end{theorem}
\noindent\emph{Proof sketch.}
Visibility counts, stopping times, maturity flags, pair horizons, frozen verdicts, and graph edges are deterministic functions of the common committed prefix. SCC condensation, topological ordering, and completion by $\kappa$ are deterministic, so replicas output the same $\prec_S$.

\begin{theorem}[Bounded Completion]
\label{thm:bounded_completion}
For any committed execution slice $S$, if the base stack eventually exports a committed frontier $R\ge T(S)$, then MRV eventually seals $\prec_S$. Moreover, for every $X\in S$,
\[
h_X\le r(X)+W_{\max}, \qquad
T(S)\le \max_{X\in S}(r(X)+W_{\max}).
\]
\end{theorem}
\noindent\emph{Proof sketch.}
The cap round $r(X)+W_{\max}$ is always included in the definition of $h_X$, so every stopping time is bounded. Once $R\ge T(S)$, all pair horizons are covered and MRV performs only finite local graph construction and deterministic completion.

\begin{theorem}[Slice Locality and Causal Consistency]
\label{thm:slice_locality_causal_consistency}
MRV outputs a total order only over AUFs in $S$. If $A,B\in S$ and $B\in\mathsf{Anc}(A)$, then $B\prec_S A$.
\end{theorem}
\noindent\emph{Proof sketch.}
The precedence graph has vertex set $S$; descendants outside $S$ are used only as evidence. MRV inserts hard causal edges for same-slice ancestry and preserves them either in the condensation-DAG order or inside SCCs before applying $\kappa$.

\begin{theorem}[Soundness of Evidence-Backed Edges]
\label{thm:evidence_edge_soundness}
Every evidence-backed edge $\mathtt{Edge}(A\to B)$ inserted into $E_S^{\mathsf{svp}}$ by MRV corresponds to $A \triangleright_{\mathrm{SVP}} B$ in the committed DAG.
\end{theorem}
\noindent\emph{Proof sketch.}
The comparator emits $\mathtt{Edge}(A\to B)$ only when both AUFs are mature, the post-coexistence window contains a positive $f+1$ visibility margin, and it contains no opposing strong signal. These are exactly the SVP conditions.

\begin{theorem}[Completeness for Structural Visibility Precedence]
\label{thm:svp_completeness}
For any pair $(A,B)$ in $S$, if $A \triangleright_{\mathrm{SVP}} B$, then MRV freezes $\mathtt{Edge}(A\to B)$ once $R\ge H(A,B)$ and includes $(A,B)$ in $E_S^{\mathsf{svp}}$ when $S$ is sealed.
\end{theorem}
\noindent\emph{Proof sketch.}
At $R\ge H(A,B)$, all rounds in the post-coexistence window are fixed. Since SVP is exactly the mature, one-sided, non-conflicting strong-signal condition tested by the comparator, MRV freezes the edge verdict and later inserts it into the sealed precedence graph.

A certified SVP edge is \emph{enforceable} if its endpoints lie in distinct SCCs of the sealed precedence graph and therefore induce a constraint in the condensation DAG.

\begin{theorem}[Conditional AUF-Level Structural Fairness]
\label{thm:conditional_structural_fairness}
For any committed execution slice $S$, MRV's final order $\prec_S$ preserves every enforceable certified SVP relation. Deterministic completion introduces no structural precedence claims.
\end{theorem}
\noindent\emph{Proof sketch.}
Any enforceable certified SVP edge appears as an edge in the condensation DAG, so every topological order places its source SCC before its destination SCC. Edges absorbed into SCCs are treated as graph-level ambiguity; completion by $\kappa$ resolves only residual choices and is not inserted into $E_S^{\mathsf{svp}}$.

\section{Evaluation}
\label{sec:evaluation}

We evaluate MRV as a post-consensus structural ordering layer for a Narwhal/Tusk-style DAG-BFT stack. The evaluation asks whether MRV preserves the throughput-oriented design point of DAG-based BFT while adding limited post-commit ordering delay and local computation. We compare three configurations: native Narwhal/Tusk, Narwhal/Tusk with MRV enabled, and a DoD-style graph-ordering reference implemented in the same experimental harness. These configurations separate three design points: fairness-agnostic deterministic traversal, MRV's post-consensus interpretation of committed DAG evidence, and explicit graph-based ordering incorporated into the system pipeline.

Our experiments vary in offered load, fault setting, and batch size. The offered-load experiment measures whether MRV sustains the native throughput/latency curve as clients increase pressure on the system. The fault-setting experiment evaluates how MRV behaves as quorum and visibility thresholds change. The batch-size experiment studies how MRV interacts with the payload granularity of DAG vertices, which affects both base dissemination and post-consensus ordering work.

\subsection{Experimental Setup}
\label{sec:eval_setup}

\textbf{Deployment.}
To match the wide-area setting targeted by DAG-based BFT protocols, we follow the geo-distributed methodology used by Narwhal/Tusk. We deploy validators on Amazon Web Services using m5.xlarge instances across five regions: N. Virginia (us-east-1), N. California (us-west-1), Sydney (ap-southeast-2), Stockholm (eu-north-1), and Tokyo (ap-northeast-1). Each instance hosts one validator. Benchmark clients submit transactions at controlled offered loads for a fixed experiment duration.

\textbf{Implementation.}
Our implementation starts from the open-source Narwhal/Tusk codebase \footnote{\url{https://github.com/asonnino/narwhal}}. In the native baseline, committed DAG outputs are delivered using the protocol's standard deterministic traversal. MRV is added as a separate ordering layer between consensus delivery and execution. Concretely, after Tusk commits a leader and exports the corresponding committed DAG output, MRV constructs the committed execution slice by selecting newly delivered AUFs, maintains the active per-slice visibility state, freezes pairwise verdicts once their horizons are reached, assembles the precedence graph, and outputs a slice-local execution order. MRV does not change Narwhal's worker dissemination, Tusk's voting logic, or the consensus commit rule; its additional work is local post-consensus processing over committed metadata. We also implement a DoD-style graph-ordering reference following DoD~\cite{nagda2025dag}. We include DoD because it is the closest order-fairness design built for a DAG-BFT setting, incorporating explicit graph-based ordering evidence into the system pipeline. Built on the same Narwhal/Tusk codebase, it changes the path before consensus delivery: workers disseminate local-order graphs through Narwhal's mempool, derive global-order graphs locally, and process those graphs in the role normally played by Tusk batches.

We report throughput as ordered transactions per second (TPS), measured after the corresponding ordering pipeline produces its output. We report end-to-end latency as the time from client submission until the transaction is output by the evaluated pipeline. For native Narwhal/Tusk, this is the latency reported by the original implementation; for MRV and the DoD-style baseline, it includes the additional ordering work performed by the corresponding layer. Unless otherwise specified, each experiment follows the Narwhal/Tusk configuration: one worker per validator, colocated benchmark clients, 512B transactions, 500KB batches, 1KB headers, a 200ms maximum batch delay, a 200ms maximum header delay, and a 50-round garbage-collection depth. Individual experiments vary the offered load, batch size, or configured fault setting from this baseline as described below.

\begin{figure}[htbp]
    \centering
    \includegraphics[width=\columnwidth]{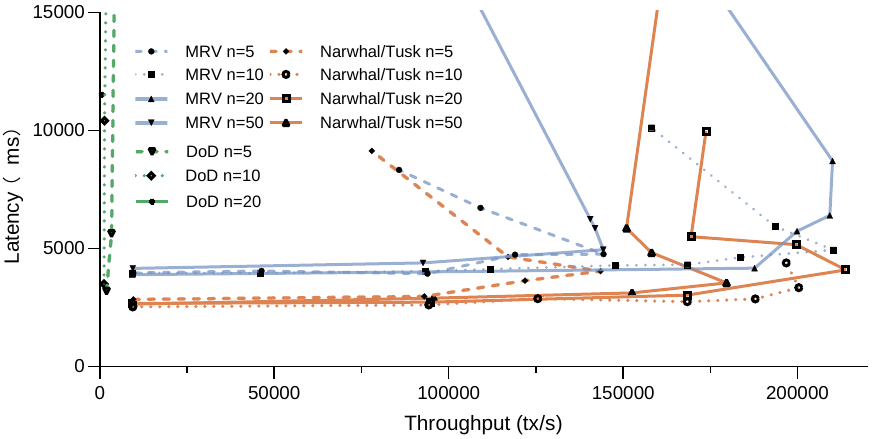}
    \caption{Comparative throughput-latency performance for MRV, Narwhal/Tusk and DoD.}
    \label{fig:ex-1}
\end{figure}

\subsection{Throughput and Latency under Offered Load}
\label{sec:eval_load}

We first measure the throughput-latency envelope under increasing offered load. Figure~\ref{fig:ex-1} reports achieved TPS and end-to-end ordering latency for native Narwhal/Tusk, Narwhal/Tusk with MRV, and the DoD-style graph-ordering reference in the same geo-distributed harness.

Across committee sizes, Narwhal/Tusk and MRV exhibit the expected saturation pattern. Throughput initially increases with offered load while latency remains relatively stable; after the knee point, additional load no longer improves throughput, and latency grows rapidly. MRV follows the same overall envelope as the native pipeline. With 5 validators, MRV reaches 144K TPS, matching the native configuration. With 10 validators, MRV reaches 210K TPS, comparable to the native peak. With 20 validators, both systems remain in the 200K TPS regime. In the 50-validator geo-distributed deployment, MRV still sustains more than 140K TPS, showing that the post-consensus ordering layer remains compatible with high-throughput DAG-BFT execution at a larger scale. The primary cost of MRV appears in latency rather than peak throughput. Before saturation, MRV adds a moderate end-to-end delay relative to native Narwhal/Tusk, reflecting bounded visibility accumulation and slice-local graph completion after commitment. This behavior is consistent with MRV's placement: ordering work is kept outside the consensus message path, but it still contributes to the final output latency observed by clients.

The DoD-style reference occupies a lower-throughput, higher-latency region in this geo-distributed harness, consistent with the cost of constructing and disseminating explicit ordering graphs through the system pipeline. MRV instead reuses committed DAG structure after consensus, so its overhead appears primarily as bounded post-commit ordering delay rather than reduced peak throughput.

\begin{figure}[htbp]
    \centering
    \includegraphics[width=\columnwidth]{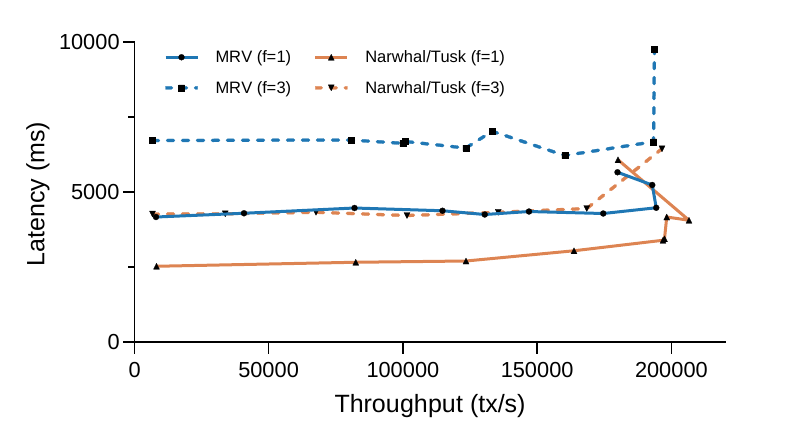}
    \caption{Throughput-latency under different configured fault-tolerance parameters.}
    \label{fig:ex-2}
\end{figure}

\subsection{Impact of Fault Setting}
\label{sec:eval_faults}

We next vary the configured fault-tolerance parameter while fixing the committee size at 10 validators. This experiment compares native Narwhal/Tusk and Narwhal/Tusk with MRV under $f=1$ and $f=3$. The parameter directly affects MRV's evidence thresholds: AUF maturity uses the $2f+1$ visibility threshold, while pairwise structural precedence uses an $f+1$ one-sided visibility margin.

Figure~\ref{fig:ex-2} shows that changing the configured $f$ has limited effect on achieved throughput in this deployment: both native Narwhal/Tusk and MRV remain near the same high-throughput region before saturation. The main effect appears in latency. Under $f=1$, MRV adds a moderate latency gap relative to native Narwhal/Tusk while maintaining a similar saturation point. Under $f=3$, the MRV curve shifts upward because stronger visibility thresholds require additional committed-DAG evidence before AUFs and pairwise verdicts can seal. The native curves provide a baseline for the underlying DAG-BFT pipeline; the additional upward shift of MRV reflects post-consensus evidence collection rather than extra consensus-path communication. Overall, in this configuration and offered-load range, changing $f$ is most visible in MRV's ordering latency. This matches MRV's design: larger $f$ raises the local evidence thresholds used for maturity and pairwise precedence, while MRV itself does not add consensus-path messages.

\begin{figure}[htbp]
    \centering
    \includegraphics[width=\columnwidth]{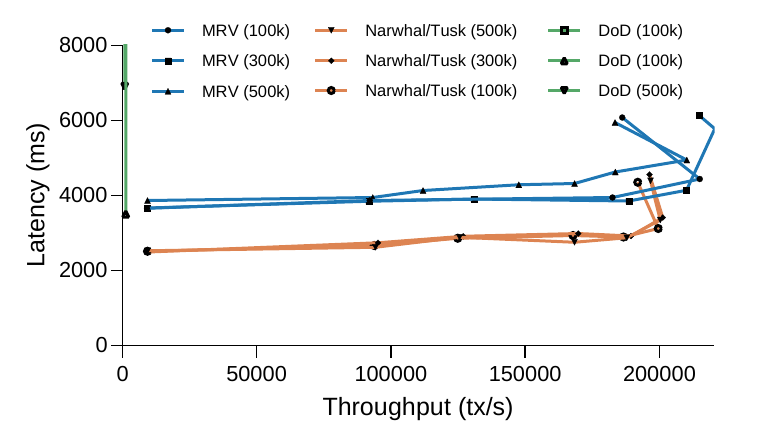}
    \caption{Throughput-latency under different batch sizes.}
    \label{fig:ex-3}
\end{figure}

\subsection{Impact of Batch Size}
\label{sec:eval_batch}

We study the effect of batch size on MRV. This experiment fixes the committee size at 10 validators and varies the batch size while measuring TPS and end-to-end ordering latency. We use this setting to isolate payload granularity: batch size changes the amount of transaction data carried by each DAG unit, while MRV's ordering work is primarily over committed AUFs and their visibility relations.

Figure~\ref{fig:ex-3} shows that MRV and native Narwhal/Tusk are insensitive to the tested batch sizes. Across 100KB, 300KB, and 500KB batches, MRV remains in the same high-throughput regime and follows the saturation pattern observed in the offered-load experiment. This suggests that MRV's post-consensus ordering cost is driven mainly by the number of committed AUFs and pairwise visibility processing, rather than by transaction payload size. The DoD-style reference remains in a lower-throughput region in the geo-distributed deployment. In this setting, wide-area dissemination and graph-processing overhead dominate the batch-size effect. The supplementary material discusses deployment differences and the role of data-dependent workload structure in interpreting our DoD-style measurements relative to DoD's published evaluation.


\section{Conclusion}
MRV is a post-consensus ordering layer that reuses committed DAG structure as execution-ordering evidence. By extracting creator-level visibility from authenticated round and ancestry metadata, MRV adds conservative AUF-level precedence constraints while bounding evidence collection and preserving causal consistency. Our prototype shows that this design preserves the high-throughput regime of the base DAG-BFT stack with limited overhead, offering a practical control point for structural ordering without moving ordering logic onto the consensus path.


\bibliographystyle{ACM-Reference-Format}
\bibliography{reference}

\end{document}